\documentclass[12pt,letterpaper]{article}
\usepackage{amsmath,amsfonts,amsthm,amssymb}

\swapnumbers 

\newtheorem{lemma}{Lemma}[section]
\newtheorem{corollary}[lemma]{Corollary}
\newtheorem{theorem}[lemma]{Theorem}
\newtheorem{example}[lemma]{Example}

\theoremstyle{definition} 
\newtheorem{definition}[lemma]{Definition}
\newtheorem{remark}[lemma]{Remark}

\newtheorem{remarks}[lemma]{Remarks}

\newcommand{\Nat}{{\mathbb N}}

\newcommand\reals{{\mathbb R}}

\newcommand{\dg}{\sp{\text{\rm o}}}

\newcommand{\spec}{\operatorname{spec}}
\newcommand{\Ext}{\operatorname{Ext}}

\begin{document}

\title{States and synaptic algebras}

\author{David J. Foulis{\footnote{Emeritus Professor, Department of
Mathematics and Statistics, University of Massachusetts, Amherst,
MA; Postal Address: 1 Sutton Court, Amherst, MA 01002, USA;
foulis@math.umass.edu.}}\hspace{.05 in} Anna Jen\v cov\'a and Sylvia
Pulmannov\'{a}{\footnote{ Mathematical Institute, Slovak Academy of
Sciences, \v Stef\'anikova 49, SK-814 73 Bratislava, Slovakia;
pulmann@mat.savba.sk. The second and third authors were supported by grant
VEGA No.2/0069/16.}}}

\date{}

\maketitle

\begin{abstract}
\noindent Different versions of the notion of a state have been
formulated for various so-called quantum structures. In this
paper, we investigate the interplay among states on synaptic
algebras and on its sub-structures. A synaptic algebra is a
generalization of the partially ordered Jordan algebra of all
bounded self-adjoint operators on a Hilbert space. The paper
culminates with a characterization of extremal states on a
commutative generalized Hermitian algebra, a special kind of
synaptic algebra.
\end{abstract}

\noindent{\bf Key Words:} synaptic algebra, GH-algebra, Jordan algebra, convex  effect
algebra, MV-algebra, $\ell$-group, order unit normed space, state, extremal state.

\medskip

\noindent{\bf AMS Classification} 81P10, 81Q10 (46B40)

\section{Introduction}

Synaptic algebras, featured in this paper, incorporate several
so-called ``quantum structures." Quantum structures were originally
understood to be mathematical systems that permit a perspicuous
representation for at least one of the key ingredients of
quantum-mechanical theory, e.g., states, observables, symmetries,
properties, and experimentally testable propositions \cite{DvPu, FEOSS,
Var}. In spite of the adjective `quantum,' a variety of mathematical
structures arising in classical physics, computer science, psychology,
neuroscience, fuzzy logic, fuzzy set theory, and automata theory are
now regarded as quantum structures.

As per the title of this paper, which is intended to complement the
authoritative articles \cite{DvurStates1, DvurStates2} by A.
Dvure\v{c}enskij, we shall study the interplay among states on
a synaptic algebra and on some of its sub-structures that qualify as
quantum structures. A state assigns to a real observable (a physical
quantity) the expected value of the observable when measured in that
state. Also, a state assigns to a testable $2$-valued proposition the
probability that the proposition will test `true' in that state.

The adjective `synaptic' is derived from the Greek word `sunaptein,' meaning
to join together; indeed synaptic algebras unite the notions of an order-unit
normed space \cite[p. 69]{Alf}, a special Jordan algebra \cite{McC}, a convex
effect algebra \cite{BBGP, BGPconvex}, and an orthomodular lattice \cite
{Beran, Kalm}.

Since virtually every quantum structure, including synaptic algebras,
are partially ordered sets (posets for short), we review some of the
basic definitions and facts concerning posets in Section \ref{sc:Posets}.
Also since every synaptic algebra $A$ is an extension of a so-called convex
effect algebra $E$, and the states on $A$ are in affine bijective
correspondence with the states on $E$, we offer a brief review of effect
algebras and states thereon in Section \ref{sc:Effect}. Moreover, every
synaptic algebra is an order-unit normed linear space, a structure
that we review in Section \ref{sc:OUNS}, where states on an order-unit
normed space are defined and some of their properties are recalled.

Section \ref{sc:SA} is devoted to a brief account of some of the basic
properties of synaptic algebras (SAs for short) and to a special case
thereof called generalized Hermitian (GH-) algebras. States on an SA are
defined just as they are for any order-unit normed space. In Section \ref
{sc:ComSA}, we consider commutative SAs and their functional representations.
In Section \ref{sc:statesCGH} we characterize extremal states on commutative
GH-algebras.

In what follows, the notation $:=$ means `equals by definition,'
the phrase `if and only if' is abbreviated as `iff,' $\Nat:=\{1,2,3,...\}$
is the system of natural numbers, the ordered field of real numbers is
denoted by $\reals$, and $\reals\sp{+}:=\{\alpha\in\reals:0\leq\alpha\}$.

\section{Partially ordered sets} \label{sc:Posets}

A binary relation $\leq$ defined on a nonempty set ${\mathcal P}$ is
a \emph{partial order relation} iff for all $a,b,c\in{\mathcal P}$,
(1) $a\leq a$, (2) $a\leq b$ and $b\leq a\Rightarrow a=b$, and (3)
$a\leq b$ and $b\leq c\Rightarrow a\leq c$. A \emph{partially ordered
set} (\emph{poset} for short) is a nonempty set ${\mathcal P}$
equipped with a distinguished partial order relation $\leq$.

Suppose that ${\mathcal P}$ is a poset with partial order relation
$\leq$ and let $a,b\in{\mathcal P}$. We write $b\geq a$ iff $a\leq b$,
and the notation $a<b$ (or $b>a$) means that $a\leq b$ but $a\not=b$.
Let ${\mathcal Q}\subseteq{\mathcal P}$. We say that $a$ is a
\emph{lower bound}, ($b$ is an \emph{upper bound}), for ${\mathcal Q}$
iff $a\leq q$, ($q\leq b$), for all $q\in{\mathcal Q}$. Also, $a$ is
the \emph{least} or the \emph{minimum}, ($b$ is the \emph{greatest} or
the \emph{maximum}) element of ${\mathcal Q}$ iff $a$ is a  lower bound
for ${\mathcal Q}$ and $a\in{\mathcal Q}$, ($b$ is an upper bound for
${\mathcal Q}$ and $b\in{\mathcal Q}$). The notation $a=\bigwedge
{\mathcal Q}$, ($b=\bigvee{\mathcal Q}$), means that $a$ is the greatest
lower bound, ($b$ is the least upper bound), of ${\mathcal Q}$. The
greatest lower bound $a=\bigwedge{\mathcal Q}$, (the least upper bound
$b=\bigvee{\mathcal Q}$), if it exists, is also called the \emph{infimum},
(the \emph{supremum}), of ${\mathcal Q}$ in ${\mathcal P}$.

If it exists, the greatest lower bound, (least upper bound), of the set
$\{a,b\}\subseteq{\mathcal P}$ is written as $a\wedge b$, (as $a\vee b$)
and is often referred to as the \emph{meet} (the \emph{join}) of $a$
and $b$. If it is necessary to make clear that an existing meet $a\wedge b$
or join $a\vee b$ is calculated in ${\mathcal P}$, it may be written as
$a\wedge\sb{\mathcal P}b$ or $a\vee\sb{\mathcal P}b$. The poset
${\mathcal P}$ is said to be \emph{lattice ordered}, or simply a \emph
{lattice} iff every pair of elements in ${\mathcal P}$ has a meet and a
join in ${\mathcal P}$. By definition, a lattice ${\mathcal P}$ is \emph
{distributive} iff, for all $a,b,c\in{\mathcal P}$, $a\wedge(b\vee c)
=(a\wedge b)\vee(a\wedge c)$, or equivalently $a\vee(b\wedge c)=
(a\vee b)\wedge(a\vee c)$. A mapping from one lattice to another is a
\emph{lattice homomorphism} iff it preserves both meets and joins.

A least (a greatest) element of the poset ${\mathcal P}$ is often written
as $0$ (as $1$). A poset with both a least and a greatest element is called
a \emph{bounded poset}. If $0$ is the least element of ${\mathcal P}$, then
the elements $a$ and $b$ in ${\mathcal P}$ are said to be \emph{disjoint}
iff, for all $c\in{\mathcal P}$, $c\leq a,b\Rightarrow c=0$, i.e., iff
$a\wedge b=0$.

Suppose that ${\mathcal P}$ is a bounded lattice. Then the elements
$a,b\in{\mathcal P}$ are said to be \emph{complements} iff $a\wedge b
=0$ and $a\vee b=1$, and ${\mathcal P}$ is called \emph{complemented}
iff every element in ${\mathcal P}$ has a complement. A \emph{Boolean
algebra} can be defined as a bounded complemented and distributive
lattice \cite{Sik}. In a Boolean algebra, every element has a unique
complement.

If $X$ is a nonempty set, then a \emph{field} of subsets of $X$ is a set
${\mathcal F}$ of subsets of $X$ such that $X\in{\mathcal F}$ and, for $S,T
\in{\mathcal F}$, $S\cap T, S\cup T, X\setminus S\in{\mathcal F}$. Partially
ordered by set containment $\subseteq$, a field ${\mathcal F}$ of subsets of
$X$ is a Boolean algebra, with $S\wedge T=S\cap T$, $S\vee T=S\cup T$, and
with $X\setminus S$ as the complement of $S$ for all $S,T\in{\mathcal F}$.

An \emph{orthocomplementation} on a bounded lattice ${\mathcal P}$ is
a mapping $a\mapsto a\sp{\perp}$ on ${\mathcal P}$ such that, for all
$a,b\in{\mathcal P}$, (1) $a$ and $a\sp{\perp}$ are complements, (2) $a
=(a\sp{\perp})\sp{\perp}$, and (3) $a\leq b\Rightarrow b\sp{\perp}\leq a
\sp{\perp}$. An \emph{orthomodular lattice} (OML) \cite{Beran, Kalm}
is a bounded lattice ${\mathcal P}$ equipped with an orthocomplementation
$a\mapsto a\sp{\perp}$ such that the \emph{orthomodular identity} $a\leq b
\Rightarrow b=a\vee(b\wedge a\sp{\perp})$ holds for all $a,b\in{\mathcal P}$.
It is well known that a Boolean algebra is the same thing as a distributive
OML.

\begin{definition} \label{df:CompletenessProps}
Let ${\mathcal P}$ be a poset. Then{\rm:}
\begin{enumerate}
\item[(1)] ${\mathcal P}$ is \emph{upward directed} (\emph{downward
 directed}) iff, for all $a,b\in{\mathcal P}$, $\{a,b\}$ has an
 upper bound (a lower bound) in ${\mathcal P}$. If ${\mathcal P}$
 is both upward and downward directed, it is said to be \emph{directed}.
\item[(2)] ${\mathcal P}$ is \emph{$\sigma$-complete} iff every
 sequence in ${\mathcal P}$ has a supremum (an infimum) in
 ${\mathcal P}$.
\item[(3)] ${\mathcal P}$ is \emph{Dedekind $\sigma$-complete} iff
 every sequence in ${\mathcal P}$ that is bounded above (below) has
 a supremum (an infimum) in ${\mathcal P}$.
\item[(4)] If $a\sb{1}\leq a\sb{2}\leq a\sb{3}\leq\cdots$ is an
 ascending sequence in ${\mathcal P}$ with supremum $a=\bigvee\sb{n=1}
 \sp{\infty}a\sb{n}$ in ${\mathcal P}$, we write $a\sb{n}\nearrow a$.
 Similar notation $a\sb{n}\searrow a$ applies to a descending sequence
 $a\sb{1}\geq a\sb{2}\geq a\sb{3}\geq\cdots$ with infimum $a=\bigwedge
 \sb{n=1}\sp{\infty}a\sb{n}$ in ${\mathcal P}$.
\item[(5)] ${\mathcal P}$ is \emph{monotone $\sigma$-complete} iff, for
 every bounded ascending (descending) sequence $(a\sb{n})\sb{n=1}\sp
 {\infty}$ in ${\mathcal P}$, there exists $a\in{\mathcal P}$ with $a
 \sb{n}\nearrow a$ (with $a\sb{n}\searrow a$).
\end{enumerate}
\end{definition}

\noindent Obviously, if ${\mathcal P}$ is $\sigma$-complete, then it is
Dedekind $\sigma$-complete, and if it is Dedekind $\sigma$-complete, then
it is monotone $\sigma$-complete.

\begin{remarks} \label{rm:involution}
Let ${\mathcal P}$ be a poset. A mapping $a\mapsto a\sp{\prime}$ from
${\mathcal P}$ to itself is called an \emph{involution} iff it is
order reversing and of period two, i.e., iff, for all $a,b\in
{\mathcal P}$, $a\leq b\Rightarrow b\sp{\prime}\leq a\sp{\prime}$ and
$(a\sp{\prime})\sp{\prime}=a$. An involution $a\mapsto a\sp{\prime}$
on ${\mathcal P}$ provides a ``duality" between upper and lower
bounds in the sense that $b$ is an upper bound for ${\mathcal Q}
\subseteq{\mathcal P}$ iff $b\sp{\prime}$ is a lower bound for $\{q
\sp{\prime}:q\in{\mathcal Q}\}$. Thus, if ${\mathcal P}$ admits an
involution, then ${\mathcal P}$ is upward directed iff it is downward
directed iff it is directed, and similar remarks apply to conditions
(2)--(5) in Definition \ref{df:CompletenessProps}.
\end{remarks}

\section{Effect algebras} \label{sc:Effect}

In this section we briefly review some definitions and facts in regard
to effect algebras, states on effect algebras, and MV-effect algebras.

\begin{definition}
An \emph{effect algebra} \cite{FoBe} is a system $(E;0,1,\sp{\perp},\oplus)$
consisting of two constants, $0$ and $1$ in $E$, a unary operation $e\mapsto
e\sp{\perp}$ on $E$, and a partially defined binary operation $\oplus$ on $E$,
called the \emph{orthosummation}, satisfying the following conditions for all
$d,e,f\in E${\rm:}
\begin{enumerate}
\item[(1)] (\emph{Associativity}) $d\oplus(e\oplus f)=(d\oplus e)\oplus f$
 in the sense that if either side is defined, then both sides are defined
 and the equality holds.
\item[(2)] (\emph{Commutativity}) $e\oplus f=f\oplus e$ in the same sense
 as in (1).
\item[(3)] (\emph{Orthosupplementation}) For each $e\in E$, there exists a
 unique element $e\sp{\perp}\in E$, called the \emph{orthosupplement of
 $e$}, such that $e\oplus e\sp{\perp}$ is defined and $e\oplus e\sp{\perp}
 =1$.
\item[(4)] (\emph{Zero-One Law}) If $e\oplus 1$ is defined, then $e=0$.
\end{enumerate}
\end{definition}

For short, we often denote an effect algebra $(E;0,1,\sp{\perp},\oplus)$
simply by $E$. Let $E$ be an effect algebra and let $d,e,f\in E$. We say
that $e$ and $f$ are \emph{orthogonal}, in symbols $e\perp f$, iff
$e\oplus f$ is defined. It is not difficult to show that the
\emph{cancelation law} holds, i.e., if $e\perp d$ and $f\perp d$,
then $e\oplus d=f\oplus d\Rightarrow e=f$.

We define the \emph{induced partial order} on $E$ by $e\leq f$ iff there
exists $d$ such that $e\perp d$ and $e\oplus d=f$. It is straightforward
to show that $E$ is a bounded poset under $\leq$. Also $e\leq f
\Leftrightarrow f\sp{\perp}\leq e\sp{\perp}$, $(e\sp{\perp})\sp{\perp}=e$,
$e\perp f\Leftrightarrow e\leq f\sp{\perp}$, and $e\leq f\Rightarrow f=
e\oplus(e\oplus f\sp{\perp})\sp{\perp}$. In particular, $e\mapsto e\sp{\perp}$
is an involution on $E$, so Remarks \ref{rm:involution} apply to $E$. If
the poset $E$ is a lattice, we say that $E$ is \emph{lattice ordered} or
that it is a \emph{lattice effect algebra} \cite{ZRBlocks}.

\begin{example} \label{ex:[0,1]}
{\rm The real closed unit interval $[0,1]\subseteq\reals$ is organized into a
lattice effect algebra by defining $e\sp{\perp}:=1-e$, and by defining
$e\perp f$ iff $e+f\leq 1$, in which case $e\oplus f:=e+f$ for all $e,f
\in[0,1]$. Then the induced partial order coincides with the restriction
to $[0,1]$ of the usual total order on $\reals$.}
\end{example}

\begin{example} \label{ex:StandardEA}
{\rm The prototype for an effect algebra is as follows: Let $\mathfrak H$
be a Hilbert space with inner product $\langle\cdot,\cdot\rangle$ and let
${\mathcal B}({\mathfrak H})$ be the algebra of bounded linear operators
on ${\mathfrak H}$. As usual, the system ${\mathcal B}\sp{sa}({\mathfrak H})$
of self-adjoint operators in ${\mathcal B}({\mathfrak H})$ is partially
ordered by $S\leq T$ iff $x\in{\mathfrak H}\Rightarrow0\leq\langle(T-S)x,x
\rangle$ for $S,T\in{\mathcal B}\sp{sa}({\mathfrak H})$. Then, denoting the
identity operator on ${\mathfrak H}$ by $1$, we define the \emph{standard
effect algebra on ${\mathfrak H}$} to be the set ${\mathcal E}({\mathfrak H})
:=\{S\in{\mathcal B}\sp{sa}({\mathfrak H}):0\leq S\leq 1\}$ with $S\sp{\perp}
:=1-S$ and $S\perp T$ iff $S+T\leq 1$, in which case $S\oplus T:=S+T$.
Verification that ${\mathcal E}({\mathfrak H})$ is an effect algebra and
that the induced partial order is the restriction to ${\mathcal E}
({\mathfrak H})$ of the partial order on ${\mathcal B}\sp{sa}({\mathfrak H})$
is straightforward. The positive-operator-valued measures featured in the
theory of quantum measurement {\rm\cite[page 9]{BLM}}, take their values in
${\mathcal E}({\mathfrak H})$.}
\end{example}

\begin{example} \label{ex:OMLEA}
{\rm Suppose that $L$ is an orthomodular lattice with orthocomplementation
$p\mapsto p\sp{\perp}$. Then $L$ is organized into an effect algebra
$(L;0,1,\sp{\perp},\oplus)$ as follows{\rm:} For $p,q\in L$, $p\perp q$
iff $p\leq q\sp{\perp}$, in which case $p\oplus q:=p\vee q$. Then the
induced partial order on $L$ coincides with the original partial order
on the OML $L$, whence $L$ is a lattice effect algebra. Moreover, if two
elements in $L$ are orthogonal, then they are disjoint.}
\end{example}

Let $E$ and $F$ be effect algebras. Then an \emph{effect-algebra morphism}
from $E$ to $F$ is a mapping $\phi\colon E\to F$ such that $\phi(1)=1$ and
for all $e,f\in E$, $e\perp f\Rightarrow \phi(e)\perp\phi(f)\text{\ and\ }
\phi(e\oplus f)=\phi(e)\oplus\phi(f)$.  An \emph{effect-algebra isomorphism}
from $E$ onto $F$ is a bijective effect-algebra morphism $\phi\colon E\to F$
such that $\phi\sp{-1}\colon F\to E$ is also an effect-algebra morphism.

Let $E$ be an effect algebra. A subset $F\subseteq E$ is called a \emph
{sub-effect algebra} of $E$ iff (1) $0\in F$, (2) $f\in F\Rightarrow f
\sp{\perp}\in F$, and (3) if $f\sb{1}, f\sb{2}\in F$ and $f\sb{1}\perp f
\sb{2}$, then $f\sb{1}\oplus f\sb{2}\in F$. As is easily seen, a sub-effect
algebra $F$ of $E$ is an effect algebra in its own right under the restriction
to $F$ of the operations on $E$, and as such, the induced partial order on $F$
is the restriction to $F$ of the induced partial order on $E$.

A triple $e\sb{1}, e\sb{2}, e\sb{3}$ in the effect algebra $E$ is said to
be \emph{orthogonal} iff $(e\sb{1}\oplus e\sb{2})\oplus e\sb{3}$ is defined,
in which case we write $\bigoplus\sb{i=1}\sp{3}e\sb{i}=e\sb{1}\oplus e\sb{2}
\oplus e\sb{3}:=(e\sb{1}\oplus e\sb{2})\oplus e\sb{3}=e\sb{1}\oplus(e\sb{2}
\oplus e\sb{3})$. Continuing in an obvious way by induction, we define \emph
{orthogonality} and the \emph{orthogonal sum} $\bigoplus\sb{i=1}\sp{n}e\sb{i}$
of a finite orthogonal sequence $(e\sb{i})\sb{i=1}\sp{n}\in E$. An infinite
sequence $(e\sb{i})\sb{i=1}\sp{\infty}$ is said to be \emph{orthogonal} iff
every finite subsequence of $(e\sb{i})\sb{i=1}\sp{\infty}$ is orthogonal. An
orthogonal sequence $(e\sb{i})\sb{i=1}\sp{\infty}\in E$ is called \emph
{orthosummable} with \emph{orthosum} $\bigoplus\sb{i=1}\sp{\infty}e\sb{i}
:=\bigvee\sb{n=1}\sp{\infty}\left(\bigoplus\sb{i=1}\sp{n}e\sb{i}\right)$ iff
the supremum exists. The effect algebra $E$ is said to be \emph
{$\sigma$-orthocomplete} iff every orthogonal sequence in $E$ is orthosummable
\cite{JPOC}. Clearly, if $E$ is monotone $\sigma$-complete, then it is
$\sigma$-orthocomplete.

\begin{definition} \label{df:EAState}
A \emph{state} on the effect algebra $E$ is a mapping $\omega\colon E
\to\reals\sp{+}$ such that (1) $\omega(1)=1$ and (2) if $e,f\in E$ and $e
\perp f$, then $\omega(e\oplus f)=\omega(e)+\omega(f)$. The set of all
states on $E$, denoted by $S(E)$, is called the \emph{state space} of
$E$. A state $\omega\in S(E)$ is said to be \emph{$\sigma$-additive} iff
for every ascending sequence $(e\sb{n})\sb{n=1}\sp{\infty}\subseteq E$,
$e\sb{n}\nearrow e\in E\Rightarrow \omega(e\sb{n})\nearrow\omega(e)$.
\end{definition}

Suppose that $\omega\in S(E)$ is $\sigma$-additive and that $(e\sb{n})\sb{n=1}
\sp{\infty}$ is an orthosummable orthogonal sequence in $E$. Then $\bigoplus
\sb{i=1}\sp{n}e\sb{i}\nearrow\bigoplus\sb{i=1}\sp{\infty}e\sb{i}$, whence
$\sum\sb{i=1}\sp{n}\omega(e\sb{i})=\omega\left(\bigoplus\sb{i=1}\sp{n}e\sb{i}
\right)\nearrow\omega\left(\bigoplus\sb{i=1}\sp{\infty}e\sb{i}\right)$,
so $\omega\left(\bigoplus\sb{i=1}\sp{\infty}e\sb{i}\right)=\sum\sb{i=1}
\sp{\infty}\omega(e\sb{i})$.

If the closed real unit interval $[0,1]\subseteq\reals$ is organized
into an effect algebra as in Example \ref{ex:[0,1]}, then a state
$\omega\in S(E)$ is the same thing as an effect-algebra morphism
$\omega\colon E\to[0,1]$. There are effect algebras with empty state
spaces \cite{GrNoStates}. If $\omega, \tau\in S(E)$ and $\lambda\in[0,1]$,
then $\lambda\omega+(1-\lambda)\tau\in S(E)$, i.e., $S(E)$ is a convex set.

\begin{definition} \label{df:Propsofstates}
For the effect algebra $E$, the set of all extreme points of $S(E)$ will
be denoted by $\Ext(S(E))$, and a state $\omega\in\Ext(S(E))$ will be
called an  \emph{extremal state}.
\end{definition}

If $L$ is an OML, then by organizing $L$ into an effect algebra as in
Example \ref{ex:OMLEA} and applying Definition \ref{df:EAState}, we
obtain the notion of a state on $L$; thus, $\omega\in S(L)$ iff
$\omega\colon L\to\reals\sp{+}$, $\omega(1)=1$, and for all $p,q\in L$,
$p\perp q\Rightarrow\omega(p\vee q)=\omega(p)+\omega(q)$. As mentioned above,
a Boolean algebra $B$ is the same thing as a distributive OML; whence
$\omega\in S(B)$ iff $\omega\colon B\to\reals\sp{+}$, $\omega(1)=1$, and
$\omega(p\vee q)=\omega(p)+\omega(q)$ whenever $p,q\in B$ and $p\wedge q=0$.
In other words, a state on $B$ is what is usually called a finitely additive
probability measure on $B$.

MV-algebras, which were originally introduced by C.C. Chang \cite{Chang}
to serve as algebraic models for many-valued logics, are mathematically
equivalent to a certain kind of effect algebra. An \emph{MV-algebra}
\cite[page 558]{JencaMV} is an algebra $(M;0,1,\sp{\perp},+)$ consisting of
two constants $0$ and $1$, a mapping $\sp{\perp}\colon M\to M$, and a
binary composition $+$ on $M$ such that, for all $x,y,z\in M$: (1) $x+
(y+z)=(x+y)+z$, (2) $x+y=y+x$, (3) $x+0=x$, (4) $(x\sp{\perp})\sp{\perp}=x$,
(5) $0\sp{\perp}=1$, (6) $x+x\sp{\perp}=1$, and (7) $x+(x+y\sp{\perp})
\sp{\perp}=y+(y+x\sp{\perp})\sp{\perp}$. For short, we usually write the
MV-algebra $(M;0,1,\sp{\perp},+)$ simply as $M$. The MV-algebra $M$ is
organized into a poset by defining $x\leq y\Leftrightarrow y=x+(x+y\sp{\perp})
\sp{\perp}$ for $x,y\in M$, and then $M$ is a bounded distributive lattice
with $x\vee y=x+(x+y\sp{\perp})\sp{\perp}$. As is well known, a Boolean algebra
$B$ is the same thing as an MV-algebra $B$ that satisfies the condition $x+x=x$
for all $x\in B$

An effect algebra $E$ is called an \emph{MV-effect algebra} iff it is
lattice ordered and every disjoint pair of elements in $E$ is an orthogonal
pair. For instance, the effect algebra $[0,1]\subseteq\reals$ in Example
\ref{ex:[0,1]} is an MV-effect algebra.

If $M$ is an MV-algebra, then $M$ is organized into an MV-effect algebra
$(M;0,1,\sp{\perp},\oplus)$ by defining $x\perp y$ iff $x\leq y\sp{\perp}$,
in which case $x\oplus y:=x+y$ for $x,y\in M$. Conversely, if $E$ is an
MV-effect algebra, then $E$ is organized into an MV-algebra by defining
$e+f:=e\oplus(e\sp{\perp}\wedge f)$ for $e,f\in E$. In this way, \emph
{MV-algebras and MV-effect algebras are mathematically equivalent} \cite
[\S 5]{FMVH}. Thus, Definition \ref{df:EAState} provides the definition of a
state on an MV-algebra.

\section{States on order-unit normed spaces} \label{sc:OUNS}

A \emph{partially ordered abelian group} (an \emph{abelian pogroup} for
short) is an abelian group $G$ equipped with a partial order relation
$\leq$ such that, for all $a,b,c\in G$, $a\leq b\Rightarrow a+c\leq
b+c$. (We use additive notation for abelian pogroups.) Let $G$ be an
abelian pogroup. Since the mapping $a\mapsto -a$ is an involution on
$G$, Remarks \ref{rm:involution} apply. If, as a poset, $G$ is a lattice,
it is called an \emph{abelian $\ell$-group}. The \emph{positive cone} in
$G$ is denoted and defined by $G\sp{+}:=\{a\in G:0\leq a\}$. If, for every
$a\in G$, the condition that $\{na:n\in \Nat\}$ is bounded above in $G$
implies that $-a\in G\sp{+}$ (i.e., $a\leq 0$), then $G$ is said to be
\emph{Archimedean}. An element $v\in G\sp{+}$ is called an \emph{order unit}
\cite[page 4]{Good} (or a \emph{strong unit}) iff, for every $a\in G$, there
exists $n\in\Nat$ such that $a\leq nv$, and an abelian pogroup with a
distinguished order unit $v$ is denoted by $(G,v)$ and referred to as a
\emph{unital abelian pogroup}. Let $(G,v)$ be a unital abelian pogroup.
Obviously, $G$ is directed. Moreover, if $a\in G$, there exist $b,c\in G
\sp{+}$ such that $a=b-c$; for instance, choose $n\in\Nat$ with $a\leq nv$,
put $b=nv$ and put $c=nv-a$.

\begin{definition} \label{df:AdditiveEct.}
Let $G$ and $H$ be partially ordered abelian pogroups and let $\xi\colon G
\to H$. Then: (1) $\xi$ is \emph{positive} iff for all $a\in G\sp{+}$, $\xi
(a)\in H\sp{+}$. (2) If $(G,v)$ and $(H,w)$ are unital abelian pogroups,
then $\xi$ is \emph{normalized} iff $\xi(v)=w$. (3) If $G\sb{1}\subseteq G$ is
closed under addition then $\xi\sb{1}\colon G\sb{1}\to H$ is \emph{additive}
iff $\xi\sb{1}(a+b)=\xi\sb{1}(a)+\xi\sb{1}(b)$ for all $a,b\in G\sb{1}$.
\end{definition}

\noindent With the notation of Definition \ref{df:AdditiveEct.}, it is
clear that $\xi\colon G\to H$ is a group homomorphism iff it is additive,
and if $\xi$ is additive and positive then it is order preserving.

A \emph{partially ordered linear space} is a linear (or vector) space
$V$ over $\reals$ such that (1) as an additive group, $V$ is an abelian
pogroup and (2) if $\alpha\in\reals\sp{+}$, then $a\leq b\Rightarrow
\alpha a\leq\alpha b$. Thus, the notions defined above for abelian
pogroups apply to a partially ordered linear space and the results
in \cite{Good} for abelian pogroups are applicable to $V$. Let $V$ be
a partially ordered linear space. If as a poset, $V$ is a lattice (i.e.,
$V$ is an abelian $\ell$-group), then $V$ is called a \emph{vector
lattice} or a \emph{Riesz space}. If $v\in V\sp{+}$ is an order unit,
then the corresponding \emph{order-unit norm} on $V$ is defined and
denoted by $\|a\|:=\inf\{0<\lambda\in\reals:-\lambda v\leq a\leq\lambda v\}$
for all $a\in V$.

An Archimedean partially ordered linear space $V$ with a distinguished
order-unit $v$ and hosting the corresponding order-unit norm $\|\cdot\|$
is called an \emph{order-unit normed space} and denoted by $(V,v)$ \cite
[pp. 67--69]{Alf}. According to \cite[Proposition II.1.2]{Alf} and \cite
[Proposition 7.12 (c)]{Good}, the order-unit norm on an order-unit normed
space $(V,v)$ satisfies the following conditions:
\[
\text{If\ }a,b\in V,\text{\ then\ }-\|a\|\leq a\leq\|a\|\text{\ and\ }-b\leq a
 \leq b\Rightarrow\|a\|\leq\|b\|.
\]

We note that the ordered field $\reals$ of real numbers forms a one-dimensional
order-unit normed space $(\reals,1)$, and the corresponding order-unit norm is
just the absolute value $|\cdot|$ on $\reals$. If $X$ is a compact Hausdorff
space, then as usual, $C(X,\reals)$ denotes the commutative associative real
linear algebra under pointwise addition, multiplication, and multiplication by
real scalars, of all continuous functions $f\colon X\to\reals$. Under pointwise
partial order, $(C(X,\reals),1)$ is an order-unit normed space, where $1$ denotes
the constant function $x\mapsto 1$; the corresponding order-unit norm coincides
with the supremum (also called the uniform) norm $\|f\|:=\sup\{|f(x)|:x\in X\}$
for all $f\in C(X,\reals)$; and with this norm, $C(X,\reals)$ is a Banach algebra.
In what follows, we always regard $\reals$ and $C(X,\reals)$ as order-unit normed
spaces $(\reals,1)$ and $(C(X,\reals),1)$, respectively.

\begin{lemma} \label{le:HomLin}
Let $(V,v)$ and $(W,w)$ be order-unit normed spaces, let $\xi\colon
V\to W$ be a positive group homomorphism. Then $\xi$ is a positive linear
mapping from $V$ to $W$.
\end{lemma}

\begin{proof}
As in the proof of \cite[Lemma 6.7]{Good}, $\xi(\alpha a)=\alpha\xi(a)$
for every rational number $\alpha$ and every $a\in A$. Let $a\in V\sp{+}$,
let $\beta\in\reals$, and let $\alpha$ and $\gamma$ be rational numbers
with $\alpha\leq\beta\leq\gamma$. Then $\alpha\xi(a)=\xi(\alpha a)\leq\xi
(\beta a)\leq\xi(\gamma a)=\gamma(\xi a)$, and since $0\leq\xi(a)$, it
follows that $-(\gamma-\alpha)\xi(a)\leq\beta\xi(a)-\xi(\beta a)\leq(\gamma-
\alpha)\xi(a)$, whence $\|\beta\xi(a)-\xi(\beta a)\|\leq\|(\gamma-\alpha)
\xi(a)\|=(\gamma-\alpha)\|\xi(a)\|$. Letting $(\gamma-\alpha)\rightarrow 0$,
we conclude that $\beta\xi(a)-\xi\beta(a)=0$, i.e., $\xi(\beta a)=\beta\xi(a)$.
Since every element in $V$ is a difference of two elements in $V\sp{+}$, it
follows that $\xi$ is linear.
\end{proof}

\begin{definition} \label{df:UnitInterval}
If $(V,v)$ is an order-unit normed space, then the ``unit interval"
$E(V,v):=\{a\in V:0\leq a\leq v\}$ is organized into a convex effect
algebra \cite{BBGP, BGPconvex} $(E(V,v);0,v,\sp{\perp},\oplus)$ as follows:
For $a,b\in E(V,v)$, $a\sp{\perp}:=v-a$ and $a\perp b$ iff $a+b\leq v$, and
then $a\oplus b:=a+b$.
\end{definition}

\noindent In Definition \ref{df:UnitInterval}, it is easy to see that the
induced partial order on $E(V,v)$ is the just the restriction of the partial
order on $V$.

Recall that a mapping from one convex set to another is said to be
\emph{affine} iff it preserves convex combinations, and an \emph{affine
isomorphism} is a bijective affine mapping with an inverse that is also
an affine mapping.

\begin{theorem} [{\rm Cf. \cite[Corollary 14.21]{DvPu}}]
\label{th:restriction/extension}
Let $(V,v)$ and $(W,w)$ be order-unit normed spaces. Then there is
an affine bijection $\xi\leftrightarrow\omega$ between positive
normalized linear mappings $\xi\colon V\to W$ and effect-algebra
morphisms $\omega\colon E(V,v)\to E(W,w)$ provided by restriction
and extension. Moreover, if $\xi\leftrightarrow\omega$, then $\xi
\colon V\to W$ is an isomorphism of order-unit normed linear
spaces iff $\omega\colon E(V,v)\to E(W,w)$ is an effect-algebra
isomorphism.
\end{theorem}

\begin{proof}
Obviously, the restriction $\omega$ of a positive normalized linear
mapping $\xi\colon V\to W$ is an effect-algebra morphism $\omega
\colon E(V,v)\to E(W,w)$. Conversely, suppose that $\omega\colon
E(V,v)\to E(W,w)$ is an effect-algebra morphism. If
$e\in E(V,v)$ and $m\in\Nat$, then $(1/m)e\in E(V,v)$ and $\omega(e)
=\omega(m(1/m)e)=m\omega((1/m)e)$, whence
\setcounter{equation}{0}
\begin{equation} \label{eq:omega1}
e\in E(V,v)\text{\ and\ }m\in\Nat\Rightarrow\omega((1/m)e)=(1/m)
\omega(e).
\end{equation}
If $e,f\in E(V,v)$ and $n,m\in\Nat$ with $ne=mf$, then $(1/m)e=(1/n)f$,
and it follows from (\ref{eq:omega1}) that $(1/m)\omega(e)=(1/n)\omega(f)$.
Thus,
\begin{equation} \label{eq:omega2}
e,f\in E(V,v),\ n,m\in\Nat,\text{\ and\ }ne=mf\Rightarrow n\omega(e)=
m\omega(f).
\end{equation}
We propose to extend $\omega$ to $\omega\sp{+}\colon V\sp{+}\to W
\sp{+}$. Indeed, suppose that $a\in V\sp{+}$ and choose $n\in\Nat$ with
$a\leq nv$. Then $e:=(1/n)a\in E(V,v)$ and we define $\omega\sp{+}(a):=
n\omega(e)\in W\sp{+}$. By (\ref{eq:omega2}), $\omega\sp{+}$ is well-defined,
and obviously it is an extension of $\omega$. If $a,b\in V\sp{+}$,
choose $n\in\Nat$ with $a+b\leq nv$. Then $a,b\leq a+b\leq nv$, whence
$\omega\sp{+}(a+b)=n\omega((1/n)(a+b))=n[\omega((1/n)a)+\omega((1/n)b)]
=n\omega(1/n)a+n\omega(1/n)b=\omega\sp{+}(a)+\omega\sp{+}(b)$.
Consequently, $\omega\sp{+}$ is additive on $V\sp{+}$.

Now we extend $\omega\sp{+}$ to $\xi\colon V\to W$. Indeed, let $a\in
V$ and write $a=b-c$ with $b,c\in V\sp{+}$. We define $\xi\colon V\to
W$ by $\xi(a)=\omega\sp{+}(b)-\omega\sp{+}(c)$. Because $\omega\sp{+}$
is additive on $V\sp{+}$, $\xi$ is well-defined and it is additive on
$V$. Obviously, $\xi$ is an extension of $\omega\sp{+}$, hence also an
extension of $\omega$. Thus $\xi$ is positive and normalized, whence by
Lemma \ref{le:HomLin}, $\xi\colon V\to W$ a positive normalized linear
mapping. The proof that the correspondence $\xi\leftrightarrow\omega$
is bijective and that it preserves convex combinations in both directions
is straightforward, as is the final statement in the theorem.
\end{proof}

\begin{definition} \label{df:StateonV}
Let $(V,v)$ be an order-unit normed space. Then a linear mapping
$\rho\colon V\to\reals$ is a \emph{state} on $(V,v)$ iff (1) it is
positive (i.e., $a\in V\sp{+}\Rightarrow\rho(a)\in\reals\sp{+}$) and
(2) it is normalized (i.e., $\rho(v)=1$). The set of all states on
$(V,v)$, called the \emph{state space of} $(V,v)$, is denoted by $S(V,v)$,
or simply by $S(V)$ if $v$ is understood. Clearly, $S(V)$ is a convex set,
and the set of extreme points of $S(V)$ is denoted by $\Ext(S(V))$. A
state $\rho\in\Ext(S(V))$ is called an \emph{extremal state}.
\end{definition}

See \cite[Corollary II.1.5 and Proposition II.1.7]{Alf} for a proof of
the following theorem.

\begin{theorem} \label{th:FnlProps}
Let $(V,v)$ be an order-unit normed space and let $\rho\colon V\to\reals$
be a nonzero linear functional on $V$.
Then{\rm:}
\begin{enumerate}
\item  If $a\in V$, then $a\in V\sp{+}$ iff $0\leq\rho(a)$ for all
 $\rho\in S(V)$.
\item If $a\in V$, then $\|a\|=\sup\{|\rho(a)|:\rho\in S(V)\}$.
\item $\rho$ is positive iff it is bounded with $\|\rho\|=\rho(v)$.
\item $\rho\in S(V)$ iff $\rho$ is bounded and $\|\rho\|=\rho(v)=1$.
\end{enumerate}
\end{theorem}

\noindent As a consequence of parts (i) and (ii) of Theorem \ref
{th:FnlProps}, the states on $(V,v)$ determine both the partial order
$\leq$ and the order-unit norm $\|\cdot\|$ on $V$. Using the Hahn-Banach
extension theorem \cite[Theorem 1.6.1]{KR} and part (iii) of Theorem
\ref{th:FnlProps}, we obtain the following corollary.

\begin{corollary} \label{co:extendstate}
Let $(V,v)$ be an order-unit normed space, suppose that $V\sb{1}$ is a
linear subspace of $V$ and that $v\in V\sb{1}$. Then, for the order-unit
normed space $(V\sb{1},v)$, we have{\rm: (i)} If $\rho\in S(V)$, then the
restriction $\rho\sb{1}$ of $\rho$ to $V\sb{1}$ is a state on $V\sb{1}$.
{\rm(ii)} Every state $\rho\sb{1}$ on $V\sb{1}$ can be extended to a state
$\rho$ on $V$.
\end{corollary}

According to the next theorem, which is an immediate consequence of Lemma
\ref{le:HomLin}, a state on an order-unit normed space $(V,v)$ is the same
thing as a state on the additive unital pogroup $(V,v)$ \cite[page 60]{Good},
whence all of the results in \cite[Chapters 4 and 6]{Good} are applicable
to $S(V,v)$. In particular, results pertaining to the topological structure
of the state space of a unital abelian pogroup apply to $S(V,v)$, but will
not be considered here.

\begin{theorem} \label{th:stateislinfnl}
If $(V,v)$ is an order-unit normed space, then a state $\rho\in S(V)$ is
the same thing as a positive normalized additive-group homomorphism $\rho
\colon V\to\reals$.
\end{theorem}

The following theorem enables us to infer properties of states on $(V,v)$
from properties of states on the effect algebra $E(V,v)$ and \emph{vice
versa}.

\begin{theorem} \label{th:rho&omega}
Let $(V,v)$ be an order-unit normed space. Then there is an affine
bijection $\rho\leftrightarrow\omega$ between states $\rho\in S(V,v)$
and states $\omega\in S(E(V,v))$ given by restriction and extension;
moreover, if $\rho\leftrightarrow\omega$, then $\rho$ is extremal
iff $\omega$ is extremal.
\end{theorem}

\begin{proof}
In Theorem \ref{th:restriction/extension}, put $(W,w):=(\reals,1)$.
Then a positive normalized linear mapping $\rho$ from $(V,v)$ to
$(W,w)$ is a state on $V$, an effect-algebra morphism $\omega$ from
$E(V,v)$ to $E(W,w)=[0,1]\subseteq\reals$ is a state on $E(V,v)$,
and we have the required affine bijection $\rho\leftrightarrow\omega$.
Because $\rho\leftrightarrow\omega$ is affine, $\rho$ is extremal
iff $\omega$ is extremal.
\end{proof}

If $X$ is a compact Hausdorff space, we recall that a \emph{probability
measure} on $X$, is a nonnegative regular Borel measure $\mu$ on $X$ such
that $\mu(X)=1$.

\begin{theorem} \label{States&ProbMeas}
Let $X$ be a compact Hausdorff space. Then{\rm:}
\begin{enumerate}
\item There is a bijective correspondence $\mu\leftrightarrow\gamma$
 between finite signed regular Borel measures $\mu$ on $X$ and bounded
 linear  functionals $\gamma$ on $C(X,\reals)$ according to $\gamma(f)=
 \int\sb{X}f\,d\mu$ for all $f\in  C(X,\reals)$.
\item For the correspondence $\mu\leftrightarrow\gamma$ in {\rm(i)},
 $\mu$ is a  probability measure on $X$ iff $\gamma$ is a state on
 $C(X,\reals)$.
\item Each $x\in X$ induces a state $\gamma\sb{x}\in S(C(X,\reals))$
 according to $\gamma\sb{x}(f):=f(x)$ for all $f\in C(X,\reals)$,
 and we have $\gamma\sb{x}(f)=\int\sb{X}f\,d\mu\sb{x}$, where $\mu
 \sb{x}$ concentrates measure on the singleton set $\{x\}$, i.e.,
 $\mu\sb{x}(\{x\})=1$.
\item $\Ext(S(C(X,\reals)))=\{\gamma\sb{x}:x\in X\}$.
\end{enumerate}
\end{theorem}

\begin{proof}
Part (i) is a consequence of the Riesz representation theorem. In view of
the fact that $\gamma\in S(C(X,\reals))$ iff $\gamma$ is positive and
normalized, (ii) follows from (i). (See \cite[page 87]{Good}.) Part (iii)
is obvious, and (iv) follows from \cite[Proposition 5.24]{Good}.
\end{proof}

\section{Synaptic and generalized Hermitian algebras} \label{sc:SA}

Synaptic algebras were introduced in \cite{Fsa}. Examples of synaptic
algebras are the self-adjoint part of a von Neumann algebra, of an
AW$\sp{*}$-algebra, or of a Rickart C$\sp{*}$-algebra. Additional examples
of synaptic algebras are: JW-algebras \cite{Top65}, AJW-algebras \cite
[\S\,20]{Top65}, JB-algebras \cite{ASS78}, and the ordered special Jordan
algebras studied by Sarymsakov, \emph{et al.} \cite{Sary}.

To help fix ideas before we proceed, the reader may consider the algebra
${\mathcal{B}}({\mathfrak H})$ of all bounded linear operators on a Hilbert
space ${\mathfrak H}$ (Example \ref{ex:StandardEA}). Then the subset
${\mathcal{B}}\sp{sa}({\mathcal{\frak H}})$ of all bounded self-adjoint
operators in ${\mathcal{B}}({\mathfrak H})$ is the prototypic example of a
synaptic algebra.

In what follows \emph{we assume that $A$ is a synaptic algebra},
(SA for short). Axioms for the SA $A$ can be found in \cite
[Definition 1.1]{Fsa}. Here we propose to review some important
features of SAs. More details can be found in \cite{Fsa, FPproj,
FPtype, FPcom, FJP2proj, FJPpande, Pid}.

Corresponding to the SA $A$ is a (real or complex) associative linear
algebra $R\supseteq A$ with unit $1$, called the \emph{enveloping
algebra} of $A$, and $A$ is a real linear subspace of $R$ with $1\in A$.
For instance, for the synaptic algebra ${\mathcal{B}}\sp{sa}({\mathcal
{\mathfrak H}})$, the enveloping algebra is ${\mathcal{B}}({\mathcal
{\mathfrak H}})$.

Equipped with a partial order relation $\leq$, $A$ is a partially ordered
real linear space, $1$ is an order-unit in $A$, and $(A,1)$ is an order-unit
normed space. We shall assume that $1\not=0$, which enables us (as usual)
to identify each real number $\lambda\in\reals$ with $\lambda1\in A$. Thus
the order-unit norm on $A$ is given by $\|a\|:=\inf\{0<\lambda\in\reals:
-\lambda\leq a\leq\lambda\}$ for $a\in A$. Limits in $A$ are understood to
be limits with respect to the order-unit norm $\|\cdot\|$.

Let $a,b\in A$. The product $ab$ is calculated in $R$ and may or may not
belong to $A$, but if $a$ and $b$ commute, then $ab=ba\in A$. If $a,b\in
A\sp{+}$, and $ab=ba$, then $ab\in A\sp{+}$. Also $a\sp{2}\in A\sp{+}$, which
makes $A$ a special Jordan algebra \cite{McC} with the Jordan product
\[
a\odot b:=\frac12(ab+ba)=\frac14[(a+b)\sp{2}-(a-b)\sp{2}]\in A.
\]
We note that $a\sp{2}=a\odot a$, $a\sp{3}=a\odot a\sp{2}$, etc., so $A$ is
closed under $a\mapsto a\sp{n}$ for all $n\in\Nat$, and therefore $A$ is
closed under the formation of real polynomials in its elements.

For $a,b\in A$, we have $aba=2a\odot(a\odot b)-a\sp{2}\odot b\in A$, and the
linear mapping $b\mapsto aba$, called the \emph{quadratic mapping} determined
by $a$, turns out to be order preserving on $A$. Moreover, if $0\leq b$, then
$aba=0\Rightarrow  ab=ba=0$; in particular, putting $b=1$, we find that $a\sp{2}
=0\Rightarrow a=0$.

For $a,b\in A$, we write $aCb$ iff $ab=ba$, and for $B\subseteq A$, we define
the \emph{commutant} and the \emph{double commutant} in the usual way:
$C(B):=\{a\in A: aCb,\ \forall b\in B\}$ and $CC(B):=C(C(B))$. We write $C(b):
=C(\{b\})$ and $CC(b):=CC(\{b\})$. If $a\in CC(b)$, we say that $a$ \emph
{double commutes} with $b$. A subset $B$ of $A$ is \emph{commutative} iff
$aCb$ for all $a,b\in B$, i.e., iff $B\subseteq C(B)$.

As is the case for any order-unit normed space, elements of the ``unit
interval" $E:=\{e\in A: 0\leq e\leq 1\}$ are called \emph{effects}, and
$E$ is organized into a convex effect algebra $(E;0,1,\sp{\perp},\oplus)$
as in Definition \ref{df:UnitInterval}. It can be shown that $E=\{e\in A:
e\sp{2}\leq e\}$.

An important sub-effect algebra of $E$ is the set $P:=\{p\in A: p\sp{2}
=p\}$ of idempotents in $A$, which are called \emph{projections}. In fact,
$P$ is the set of extreme points of $E$ \cite[Theorem 2.6]{Fsa} and under
the restriction to $P$ of $\leq$ on $A$, $P$ is an orthomodular lattice
(OML) with orthocomplementation $p\mapsto p\sp{\perp}=1-p$. It is
easily seen that, if $p,q\in P$, then $p\oplus q$ is defined in $E$
(i.e., $p+q\leq 1$) iff $pq=qp=0$, in which case $p\oplus q=p+q=p\vee q$.

If $p\in P$ and $e\in E$, then $e\leq p$ iff $e=pep=pe=ep$, and $p\leq e$
iff $p=pe=ep=pep$, and in either case $pCe$. It turns out that the
infimum $p\wedge q$ and the supremum $p\vee q$ of projections $p$ and
$q$ in the OML $P$ are also the infimum and supremum, respectively, of
$p$ and $q$ in the effect algebra $E$. Therefore, we can safely denote
by $e\vee f$ and $e\wedge f$ existing suprema and infima in $E$ of
effects $e$ and $f$. If $e\in E$, $p\in P$, and $eCp$, then $e\wedge p$
exists in $E$ and $e\wedge p=pe=ep$. The question of exactly when two
effects $e$ and $f$ have an infimum or a supremum in $E$ is important
and open (cf. \cite{GGJ}).

Every $a\in A\sp{+}$, has a \emph{square root} $a\sp{1/2}\in A\sp{+}$
which is uniquely determined by the condition $(a\sp{1/2})\sp{2}=a$;
moreover, $a\sp{1/2}\in CC(a)$ \cite[Theorem 2.2]{Fsa}.

Let $a\in A$. Then $a\sp{2}\in A\sp{+}$, and the \emph{absolute value}
of $a$, defined by $|a|:=(a\sp{2})\sp{1/2}$, has the property that $|a|
\in CC(a)$. Moreover, $|a|$ is uniquely determined by the properties
$|a|\in A\sp{+}$ and $|a|\sp{2}=a\sp{2}$. Furthermore, using the
absolute value of $a$, we define the \emph{positive part} and the \emph
{negative part} of $a$ by $a\sp{+}:=\frac12(|a|+a)\in A\sp{+}\cap CC(a)$
and $a\sp{-}:=\frac12(|a|-a)\in A\sp{+}\cap CC(a)$, respectively. Then
$a=a\sp{+}-a\sp{-}$, $|a|=a\sp{+}+a\sp{-}$, and $a\sp{+}a\sp{-}=a\sp{-}
a\sp{+}=0$.

If $a\in A$, there exists a unique projection, denoted by $a\dg\in P$
and called the \emph {carrier} of $a$, such that, for all $b\in A$,
$ab=0\Leftrightarrow a\dg b=0$. It turns out that $ab=0\Leftrightarrow
a\dg b\dg=0\Leftrightarrow b\dg a\dg=0\Leftrightarrow ba=0$. By \cite
[Theorem 2.10]{Fsa}, $a\dg\in CC(a)$, $|a|\dg=a\dg$, and $(a^n)\dg=
a\dg$ for all $n\in\Nat$. For the SA ${\mathcal B}\sp{sa}({\mathfrak H})$,
the carrier of a self-adjoint operator $T$ is the (orthogonal) projection
onto the closure of the range of $T$.

An element $a\in A$ is \emph{invertible} iff there is a (necessarily
unique) element $a\sp{-1}\in A$ such that $aa\sp{-1}=a\sp{-1}a=1$. By
\cite[Theorem 7.2]{Fsa}, $a$ is invertible iff there exists $0\leq
\epsilon\in\reals$ such that $\epsilon\leq|a|$.

A linear subspace $V$ of $A$ is called a \emph{sub-synaptic algebra} of
$A$ iff, for all $a\in V$, (1) $1\in V$, (2) $a\sp{2}\in V$, (3) $0\leq a
\Rightarrow a\sp{1/2}\in V$, (4) $a\dg\in V$, and (5) $1\leq a\Rightarrow
a\sp{-1}\in V$. Under the partial order and the operations inherited from
$A$, a sub-synaptic algebra of $A$ is a synaptic algebra in its own right
\cite[Theorem 2.7]{FPproj}.

If $B\subseteq A$, then the commutant $C(B)$ is a norm closed sub-synaptic
algebra of $A$, and if $B$ is commutative, then $CC(B)$ is a norm closed
commutative sub-synaptic algebra of $A$ with $B\subseteq CC(B)$. In particular,
for all $a\in A$, $CC(a)$ is a norm closed commutative sub-synaptic algebra of
$A$ and $a\in CC(a)$, whence $A$ is covered by commutative norm closed
sub-synaptic algebras. Also the commutant $C(A)$ of $A$ itself, which is
called the \emph{center} of $A$, is a norm closed commutative sub-synaptic
algebra of $A$.

By \cite[Section 8]{Fsa}, each element $a\in A$ both determines and
is determined by a corresponding \emph{spectral resolution}, namely
the ascending family $(p\sb{a,\lambda})\sb{\lambda\in\reals}$ of
projections in $P\cap CC(a)$ given by $p\sb{a,\lambda}:=1-((a-\lambda)
\sp{+})\dg$ for all $\lambda \in {\mathbb R}$. It turns out that
$a=\int\sb{L\sb{a}-0}\sp{U\sb{a}}\lambda dp\sb{a,\lambda}$, where the
Riemann-Stiltjes type integral converges in norm, and where $L\sb{a}:=\sup
\{\lambda\in\reals:p\sb{a,\lambda}=0\}\in\reals$ and $U\sb{a}:=\inf
\{\lambda\in\reals: p\sb{a,\lambda}=1\}\in\reals$ are the so-called
\emph{spectral bounds} of $a$. Furthermore, by \cite[Theorem 8.10]{Fsa},
if $b\in A$, then $bCa$ iff $bCp\sb{a,\lambda}$ for all $\lambda\in\reals$;
hence $C(A)=C(P)$, and two elements in $A$ commute iff the projections in
their respective spectral resolutions commute.

The \emph{resolvent set} of $a\in A$ is the set of all $\mu\in\reals$
such that there is an open interval $I$ in $\reals$ with $\mu\in I$ and
$p\sb{a,\lambda}=p\sb{a,\mu}$ for all $\lambda\in I$. The \emph{spectrum}
of $a\in A$, denoted by $\spec(a)$, is defined as the complement in
$\reals$ of the resolvent set of $a$. By \cite[Lemma 3.1]{FPproj}, $\spec(a)
=\{\lambda\in\reals: a-\lambda\ \mbox{is not invertible}\}$. If $a\in A$,
then $\spec(a)$ is a nonempty closed and bounded subset of $\reals$, $\|a\|
=\sup\{|\lambda|:\lambda\in\spec(a)\}$, $U\sb{a}=\sup\, \spec(a)$, and
$L\sb{a}=\inf\, \spec(a)$. Moreover, $a\in A\sp{+}\Leftrightarrow\spec(a)
\subseteq\reals\sp{+}$, $a\in P\Leftrightarrow\spec(a)\in\{0,1\}$, and
$a\in E\Leftrightarrow\spec(a)\subseteq [0,1]$.

An element $a\in A$ is \emph{simple} iff it can be written as a finite
linear combination of pairwise commuting projections. By \cite[Theorem 8.9]
{Fsa}, $a$ is simple iff $\spec(a)$ is finite, and if $a$ is simple, it
can be written uniquely in \emph{spectral form} as $a=\sum\sb{i=1}\sp{n}
\lambda\sb{i}p\sb{i}$ with $\lambda\sb{1}<\lambda\sb{2}<\cdots<\lambda
\sb{n}$, $0\neq p\sb{i}\in P$, and $\sum\sb{i=1}\sp{n}p\sb{i}=1$. Then
$p\sb{i}p\sb{j}=0$ for $1\leq i\not=j\leq n$ and $\spec(a)=\{\lambda
\sb{1},\lambda\sb{2},\ldots, \lambda\sb{n}\}$. In fact, $\lambda\sb{i}$
are the \emph{eigenvalues} of $a$ and $p\sb{i}$ are the corresponding
pairwise orthogonal \emph{eigenprojections} of $a$ for $i=1,2,...,n$,
\cite[Definition 8.2 (iii)]{Fsa}. If $a$ is simple with spectral form
$a=\sum\sb{i=1}\sp{n}\lambda\sb{i}p\sb{i}$ and $f$ is a real polynomial
function, it is not difficult to see that $f(a)=\sum\sb{i=1}\sp{n}
f(\lambda\sb{i})p\sb{i}$.

We shall denote by $A\sb{0}$ the set of all simple elements in $A$. By
\cite[Corollary 8.6]{Fsa}, each element $a\in A$ is the norm limit and
also the supremum of an ascending sequence of pairwise commuting elements
in $A\sb{0}$, each of which is a finite linear combination of projections
belonging to the spectral resolution of $a$. In particular, $A\sb{0}$ is
norm dense in $A$.

The following theorem follows \emph{mutatis mutandis} from the proof of
\cite[Proposition 3.9]{Hand}.

\begin{theorem} \label{th:NormComplete}
If $A$ is a monotone $\sigma$-complete, then $A$ is norm complete, i.e.,
$A$ is a Banach synaptic algebra.
\end{theorem}

\begin{definition} \label{df:Commutative VP}
$A$ has the \emph{commutative Vigier property} iff, for every bound\-ed
ascending sequence $(a\sb{n})\sb{n=1}\sp{\infty}$ of pairwise commuting
elements in $A$, there exists $a\in CC(\{a\sb{n}:n\in\Nat\})$ with
$a\sb{n}\nearrow a$.
\end{definition}

A so-called \emph{generalized Hermitian} (GH-) \emph{algebra},
which was defined and studied in \cite[Definition 2.1]{FPgh} and
\cite{FPreg}, turns out to be the same thing as a synaptic algebra
that satisfies the commutative Vigier property \cite[\S 6]{Fsa}.
By \cite[Lemma 5.4]{FPgh}, if $A$ is a GH-algebra, then $P$ is a
$\sigma$-complete OML. Clearly, if $A$ is commutative, then $A$ is
a GH-algebra iff $A$ is monotone $\sigma$-complete. Thus, by
Theorem \ref{th:NormComplete}, a commutative GH-algebra is norm complete.

Every synaptic algebra of finite rank (meaning that there exists
$n\in\Nat$ such that there are $n$, but not $n+1$, mutually orthogonal
nonzero projections in $P$) is a GH-algebra. According to \cite{FPSpin},
a synaptic algebra of rank 2 is the same thing as a spin factor of
dimension greater than $1$. Thus, synaptic algebras of finite rank
need not be finite dimensional.

Condition (4) in the following definition is an enhancement of
\cite[Definition 2.9]{FPproj} but this does not affect the definition
of a synaptic isomorphism.

\begin{definition} \label{df:morphisms}
Let $A\sb{1}$ and $A\sb{2}$ be
synaptic algebras. A linear mapping $\phi\colon A\sb{1}\to A\sb{2}$
is a \emph{synaptic morphism} iff, for all $a,b\in A\sb{1}$:
\[
\begin{array}{ll}
 {\rm(1)}\  \phi(1)=1. &\ \ \ \  {\rm(2)}\ \phi(a\sp{2})=\phi(a)\sp{2}. \\
 {\rm(3)}\ aCb\Rightarrow\phi(a)C\phi(b). &\ \ \ \ {\rm(4)}\ \phi(a\dg)=
 \phi(a)\dg.
\end{array}
\]

\noindent A synaptic morphism $\phi$ is a \emph{synaptic isomorphism}
iff it is a bijection and $\phi\sp{-1}$ is also a synaptic morphism.
A synaptic morphism $\phi$ is a \emph{GH-algebra morphism} iff, for
every ascending sequence of mutually commuting elements $(a\sb{n})
\sb{n=1}\sp{\infty}$ in $A$,
\[
{\rm(5)}\ a\sb{n}\nearrow a\in A\Rightarrow \phi(a\sb{n})\nearrow
 \phi(a).
\]
\end{definition}

\begin{remark} \label{rm:GHisomorphism}
Since a synaptic isomorphism is necessarily an order isomorphism,
if $A\sb{1}$ is a GH-algebra, $A\sb{2}$ is an SA, and there is a
synaptic isomorphism $\phi$ from $A\sb{1}$ onto $A\sb{2}$, then
$A\sb{2}$ is a GH-algebra and $\phi$ is a GH-isomorphism.
\end{remark}

A \emph{state} on the synaptic algebra $A$ is defined just as it
is for any order-unit normed space, namely as a linear functional
$\rho\colon A\to\reals$ that is positive ($a\in A\sp{+}\Rightarrow
\rho(a)\in\reals\sp{+}$) and normalized ($\rho(1)=1$) (Definition
\ref{df:StateonV}). The state space of $A$ and the set of extremal
states on $A$ are denoted by $S(A)$ and $\Ext(S(A))$. Likewise,
$S(E)$, $\Ext(S(E))$, $S(P)$, and $\Ext(S(P))$ denote the states
and extremal states on the convex effect algebra $E\subseteq A$
and on the OML $P\subseteq E$. Thus, Theorems \ref{th:FnlProps},
\ref{th:stateislinfnl}, \ref{th:rho&omega}, and Corollary \ref
{co:extendstate} hold for states on $A$ and $E$. In particular,
$S(A)$ determines the partial order and the order-unit norm on $A$.
Moreover, there is an affine bijection $\rho\leftrightarrow\omega$
between states $\rho\in S(A)$ and states $\omega\in S(E)$ given by
restriction and extension, and $\rho$ is extremal iff $\omega$ is
extremal.

\section{Commutative synaptic algebras}  \label{sc:ComSA}

Clearly, the synaptic algebra $A$ is \emph{commutative} iff $A=C(A)$,
i.e., iff $A$ equals its own center. A commutative synaptic algebra is
a commutative associative partially ordered Archimedean real linear
algebra with a unity element $1$ that is also an order unit; it is a
normed linear algebra under the order-unit norm; and it can be regarded
as its own enveloping algebra. If $A$ is commutative, then the set
$A\sb{0}$ of simple elements in $A$ is a norm dense commutative
sub-synaptic algebra of $A$.

\begin{remark} \label{rm:ComIsom}
Using the axioms for an SA \cite[Definition 1.1]{Fsa}, it is not
difficult to see that if $A\sb{1}$ is a commutative SA and $A\sb{2}$
is a real associative algebra , then $\phi\colon A\sb{1}\to A\sb{2}$
is an algebra isomorphism of $A\sb{1}$ onto $A\sb{2}$ iff $A\sb{2}$
is a commutative SA and $\phi$ is a synaptic isomorphism of $A\sb{1}$
onto $A\sb{2}$.
\end{remark}

The following lemma is a consequence of \cite[Theorem 5.12]{vectlat} and
\cite[\S 6]{Fsa}.

\begin{lemma} \label{lm:ComLat}
The following conditions are mutually equivalent{\rm:}\newline
{\rm(i)} $A$ is commutative. {\rm(ii)} $A$ is lattice ordered, hence
 a vector lattice. {\rm(iii)} $E$ is an MV-effect algebra. {\rm(iv)}
$P$ is a Boolean algebra.

\noindent Moreover, every Boolean algebra can be realized as the Boolean
algebra of projections in a commutative synaptic algebra.
\end{lemma}

Functional representations for a commutative synaptic algebra and for
a commutative GH-algebra involve the Archimedean partially ordered
real commutative associative Banach algebra $C(X,\reals)$ corresponding
to a compact Hausdorff space $X$.  As was observed in \cite[p. 244]{FPproj},
$C(X,\reals)$ satisfies all of the synaptic algebra axioms, with the possible
exception of the existence of carriers (Axiom SA6 in \cite{Fsa}).

Let $X$ be a compact Hausdorff space. The ``unit interval" $E(X,\reals):=
\{e\in C(X,\reals):0\leq e\leq 1\}$ is organized into a convex effect
algebra just as it is for any order-unit normed space (Definition \ref
{df:UnitInterval}). We define ${\mathcal F}(X)$ to be the set of all compact
open (clopen) subsets of $X$. Then ${\mathcal F}(X)$ is a field of subsets of
$X$, whence, partially ordered by set containment $\subseteq$, it is a Boolean
algebra. We also define $P(X,\reals):=\{p\in C(X,\reals):p=p\sp{2}\}$ and note
that $P(X,\reals)$ is a sub-effect algebra of $E(X,\reals)$ consisting of all
characteristic set functions (indicator functions) $\chi\sb{K}$ of compact
open sets $K\in{\mathcal F}(X)$. Thus, $P(X,\reals)$ is a Boolean algebra
isomorphic to ${\mathcal F}(X)$ under the mapping $\chi\sb{K}\mapsto K$.
Moreover, we define $C(X,\reals)\sb{0}$ to be the set of all finite linear
combinations of elements in $P(X,\reals)$. Then $C(X,\reals)\sb{0}$ is a
subalgebra of $C(X,\reals)$, $1\in C(X,\reals)\sb{0}$, and it is not difficult
to confirm that $C(X,\reals)\sb{0}$ is a commutative synaptic algebra. What
about $C(X,\reals)$? The next theorem answers this question.

\begin{theorem} \label{th:C(X,reals)Props}
Let $X$ be a compact Hausdorff space. Then the following conditions are
mutually equivalent{\rm:}
\begin{enumerate}
\item $C(X,\reals)$ is a GH-algebra.
\item $C(X,\reals)$ is a synaptic algebra.
\item $C(X,\reals)$ has the Rickart property, i.e., if $f\in C(X,
 \reals)$, then there exists $p\in P(X,\reals)$ such that for all
 $g\in C(X,\reals)$, $fg=0\Leftrightarrow g=pg$.
\item $X$ is basically disconnected, i.e., the closure of any open
 F$\sb{\sigma}$ subset of $X$ remains open.
\item $P(X,\reals)$ is a $\sigma$-complete Boolean algebra.
\item $C(X,\reals)$ is Dedekind $\sigma$-complete.
\item $C(X,\reals)$ is monotone $\sigma$-complete.
\end{enumerate}
\end{theorem}

\begin{proof} Leaving (i) aside for a moment, we begin by proving
that (ii) $\Leftrightarrow$ (iii) $\Leftrightarrow$ (iv). As we observed
above, (ii) holds iff each $f\in C(X,\reals)$ has a carrier $f\dg\in
P(X,\reals)$ satisfying the condition that, for all $g\in C(X,\reals)$,
$fg=0\Leftrightarrow f\dg g=0$. But with $p:=1-f\dg$, the latter
condition is equivalent to $fg=0\Leftrightarrow g=pg$, and we have
(ii) $\Leftrightarrow$ (iii) (Cf. \cite[Prop. 1]{Berb}).  Using Urysohn's
lemma, it is not difficult to see that (iii) holds iff, for all $f\in
C(X,\reals)$, the closure ${\bar S}$ of $S:=\{x\in X: f(x)\neq 0\}$, is
open in $X$, and by \cite[Theorem C, p. 217]{Hal}, the latter condition
is equivalent to (iv).

(iv) $\Leftrightarrow$ (v). It is well-known that $X$ is basically
disconnected iff the field ${\mathcal F}(X)$ of compact open subsets of
$X$ is a $\sigma$-complete Boolean algebra. Since the Boolean algebra
$P(X,\reals)$ is isomorphic to ${\mathcal F}(X)$, we have (iv)
$\Leftrightarrow$ (v).

As a consequence of \cite[Theorem 5.7]{FPreg}, (i), (iv), (vi),
and (vii) are mutually equivalent, whence (i)--(vii) are
mutually equivalent.
\end{proof}

A \emph{Stone space} is defined to be a compact Hausdorff space $X$ such
that the compact open sets $K\in{\mathcal F}(X)$ form a basis for the
open sets in $X$. If $X$ is a Stone space, then $C(X,\reals)\sb{0}$
separates points in $X$, whence by the Stone-Weierstrass theorem,
the commutative synaptic algebra $C(X,\reals)\sb{0}$ is norm dense
in $C(X,\reals)$.

\begin{remarks} \label{rm:Stone}
Let $B$ be a Boolean algebra, let ${\bf 2}:=\{0,1\}$ be the two-element
Boolean algebra, and let $X$ be the set of all Boolean homomorphisms $x
\colon B\to{\bf 2}$. With the discrete topology on ${\bf 2}$, the product
space ${\bf 2}\sp{B}$ is compact and Hausdorff, and $X\subseteq{\bf 2}
\sp{B}$ is given the corresponding subspace topology. Then $X$ is a
Stone space called  \emph{the Stone space} of the Boolean algebra $B$.
According to M.H. Stone's representation theorem for Boolean algebras,
for each $b\in B$, $K\sb{b}:=\{x\in X:x(b)=1\}\in{\mathcal F}(X)$ and
the mapping defined by $b\mapsto K\sb{b}$ is a Boolean isomorphism of $B$
onto ${\mathcal F}(X)$. Therefore, the mapping $\psi\colon B\to P(X,\reals)$
defined by $\psi(b):=\chi\sb{K\sb{b}}$ is a Boolean isomorphism of $B$ onto
$P(X,\reals)$.
\end{remarks}

We note that if any one, hence all of the conditions in Theorem
\ref{th:C(X,reals)Props} hold, then $X$ is a Stone space and each
$x\in X$ can be identified with a Boolean homomorphism $x\colon P(X,
\reals)\to{\bf 2}$ according to $x(\chi\sb{K}):=\chi\sb{K}(x)$ for
every $K\in{\mathcal F}(X)$. In this way, $X$ can be regarded as the
Stone space of the $\sigma$-complete Boolean algebra $P(X,\reals)$.

If $A$ is commutative, then, unless $A$ is a GH-algebra, by
Remark \ref{rm:ComIsom} and Theorem \ref{th:C(X,reals)Props},
there cannot be an algebra isomorphism from $A$ onto $C(X,
\reals)$; however, by \cite[Theorem 4.1]{FPproj}, we do have
a functional representation for $A$ as follows.

\begin{theorem} \label{th:F}
Suppose that $A$ is commutative and let $X$ be the Stone space of
the Boolean algebra $P$. Then there is a subalgebra $F$ of
$C(X,\reals)$ such that $F$ is a synaptic algebra, $C(X,\reals)
\sb{0}$ is a sub-synaptic algebra of $F$, $C(X,\reals)\sb{0}$
is norm dense in $C(X,\reals)$, and there is a synaptic isomorphism
$\Psi$ of $A$ onto $F$ such that the restriction $\psi$ of $\Psi$ to
$P$ is a Boolean isomorphism of $P$ onto $P(X,\reals)$ as per
Stone's representation theorem {\rm(Remarks \ref{rm:Stone})}.
\end{theorem}

For a commutative GH-algebra, we have the following functional
representation theorem (see \cite[Theorem 5.9]{FPreg} for an
alternative proof).

\begin{theorem} \label{th:FnRepGH}
Suppose that $A$ is a commutative GH-algebra and let $X$ be the
basically disconnected Stone space of the $\sigma$-complete Boolean
algebra $P$. Then there is a synaptic isomorphism $\Psi\colon A\to
C(X,\reals)$ of $A$ onto $C(X,\reals)$ such that the restriction
$\psi$ of $\Psi$ to $P$ is a Boolean isomorphism of $P$ onto
$P(X,\reals)$ as per Stone's representation theorem {\rm(Remarks
\ref{rm:Stone})}.
\end{theorem}

\begin{proof}
By Theorem \ref{th:F}, there is a synaptic isomorphism $\Psi$ of
$A$ onto the norm dense synaptic algebra $F$ of $C(X,\reals)$, and
since $A$ is monotone $\sigma$-complete, so is $F$. Thus, by
Theorem \ref{th:NormComplete}, $F$ is norm complete, whence $F
=C(X,\reals)$. Thus by Remarks \ref{rm:ComIsom}, $C(X,\reals)$
is a GH-algebra and $\Psi$ is a synaptic isomorphism.
\end{proof}

\begin{theorem}\label{th:CommutGH} Let $A$ be a commutative synaptic
algebra. Then the following conditions are mutually equivalent{\rm:}
\begin{enumerate}
\item $A$ is a GH-algebra.
\item There exists a compact Hausdorff space $Y$ such that, as a
 real associative algebra, $A$ is isomorphic to $C(Y,\reals)$.
\item $A$ is Dedekind $\sigma$-complete.
\item $A$ is monotone $\sigma$-complete.
\item The unit interval $E$ of $A$ is an MV-algebra that as a lattice
 is $\sigma$-complete.
\item $A$ is norm complete.
\end{enumerate}
Moreover, any of conditions {\rm{(i)--(vi)}} implies that the projection
lattice $P$ of $A$ is a $\sigma$-complete Boolean algebra.
\end{theorem}

\begin{proof} (i) $\Leftrightarrow$ (ii). That (i) $\implies$ (ii) follows
from Theorem \ref{th:FnRepGH}. Conversely, if (ii) holds, then by Remark
\ref{rm:ComIsom}, $C(Y,\reals)$ is an SA and there is a synaptic isomorphism
of $A$ onto $C(Y,\reals)$. Thus, by Theorem \ref{th:C(X,reals)Props},
$C(Y,\reals)$ is a GH-algebra, and by Remark \ref{rm:GHisomorphism}, (i) holds.

(i) $\Rightarrow$ (iii) $\Rightarrow$ (iv) $\Rightarrow$ (i). Assume (i).
Then by Theorem \ref{th:FnRepGH}, there is a synaptic isomorphism of $A$
onto $C(X,\reals)$ where $X$ is a compact Hausdorff and basically
disconnected space, by Theorem \ref{th:C(X,reals)Props}, $C(X,\reals)$
is Dedekind $\sigma$-complete, whence $A$ is also Dedekind $\sigma$-complete,
so (i) $\Rightarrow$ (iii). Clearly, (iii) $\Rightarrow$ (iv), and since
$A$ is commutative, (iv) $\Rightarrow$ (i). Thus, (i)--(iv) are mutually
equivalent.

(iii) $\implies$ (v). Assume (iii). By Lemma \ref{lm:ComLat}, $E$ is an
MV-(effect)algebra, $E$ is bounded below by $0$ and above by $1$, the
mapping $e\mapsto 1-e$ is an involution on $E$, and $E$ inherits Dedekind
$\sigma$-completeness from $A$, whence $E$ is a $\sigma$-complete lattice,
and we have (v).

(v) $\Rightarrow$ (ii). Assume (v). Note that $E$ is a convex subset of
$A$, so by \cite[Theorem 7.3.9]{DvPu}, there is a compact Hausdorff space
$Y$ and an effect-algebra isomorphism $\phi$ from $E$ onto the unit interval
$E(Y,\reals)$ of $C(Y,\reals)$. Evidently, $\phi$ is a lattice isomorphism.
By Theorem \ref{th:restriction/extension}, $\phi$ can be extended to a
positive normalized linear isomorphism, $\Phi\colon A\to C(Y,\reals)$ which
is obviously an isometry. If $p\in E$, then by \cite[Theorem 2.6]{Fsa},
$p\in P\Leftrightarrow p\wedge(1-p)=0$, and since the operations on $C(Y,
\reals)$ are performed pointwise, a similar condition holds for the elements
of $E(Y,\reals)$ that belong to $P(Y,\reals)$. Therefore, $\Phi$ maps $P$
onto $P(Y,\reals)$, and as such it is a Boolean isomorphism.

Suppose $a\in A\sb{0}$ with $a=\sum\sb{i=1}\sp{n}\lambda\sb{i}p\sb{i}$ in
spectral form. Then $\Phi(a)=\sum\sb{i=1}\sp{n}\lambda\sb{i}\Phi(p\sb{i})
\in C(Y,\reals)\sb{0}$, and since $a\sp{2}=\sum\sb{i=1}\sp{n}\lambda\sb{i}
\sp{2}p\sb{i}$, it follows that $\Phi(a\sp{2})=\sum\sb{i=1}\sp{n}\lambda
\sb{i}\sp{2}\Phi(p\sb{i})=\Phi(a)\sp{2}$.

If $a\in A$, then $a$ is a norm limit $a\sb{n}\rightarrow a$ of a sequence
$(a\sb{n})\sb{n=1}\sp{\infty}\in A\sb{0}$. Therefore, since $\Phi$ is an
isometry and both $A$ and $C(Y,\reals)$ are commutative, we have $a\sb{n}
\sp{2}\rightarrow a\sp{2}$, $\Phi(a\sb{n})\sp{2}\rightarrow\Phi(a)\sp{2}$,
and $\Phi(a\sb{n})\sp{2}=\Phi(a\sb{n}\sp{2})\rightarrow \Phi(a\sp{2})$,
whence $\Phi(a\sp{2})=\Phi(a)\sp{2}$, i.e., $\Phi$ preserves squares.
Therefore, if $a,b\in A$, then since $A$ is commutative, $ab=\frac12\left
((a+b)\sp{2}-a\sp{2}-b\sp{2}\right)$, it follows that $\Phi(ab)=\Phi(a)
\Phi(b)$, and we have (ii).  Thus, (i)--(v) are mutually equivalent.

(ii) $\Leftrightarrow$ (vi). That (ii) $\Rightarrow$ (vi) follows from
the observation that the algebra isomorphism in (ii) is necessarily
an isometry. Conversely, assume (vi). Then, with the notation of Theorem
\ref{th:F}, there is a synaptic isomorphism $\Psi$ of $A$ onto $F
\subseteq C(X,\reals)$. Also since $C(X,\reals)\sb{0}\subseteq F$ is norm
dense in $C(X,\reals)$, so is $F$. By (vi), $A$ is norm complete,
hence $F$ is norm complete as well; therefore $F=C(X,\reals)$, and
with $Y:=X$, we have (ii).

The remaining statement follows by \cite[Theorem 5.7]{FPreg}.
\end{proof}

\section{States on a commutative GH-algebra} \label{sc:statesCGH}

By Lemma \ref{lm:ComLat}, Theorem \ref{th:stateislinfnl} and \cite
[Theorem 12.18]{Good}, we have the following result.

\begin{theorem} \label{th:extreme}
If $A$ is a commutative synaptic algebra and $\rho\in S(A)$, then
$A$ is a vector lattice and the following conditions are mutually
equivalent{\rm: (i)} $\rho\in\Ext(S(A))$. {\rm(ii)} $\rho\colon A
\to\reals$ is a lattice homomorphism. {\rm(iii)} $\rho(a\wedge\sb{A}b)
=\min\{\rho(a),\rho(b)\}$ for all $a,b\in A\sp{+}$.
\end{theorem}

\begin{theorem}
Suppose that $A$ is a commutative GH-algebra, $X$ is the compact
Hausdorff basically disconnected Stone space of the $\sigma$-complete
Boolean algebra $P$, $\Psi\colon A\to C(X,\reals)$ is the synaptic
isomorphism of Theorem \ref{th:FnRepGH}, and $\rho\in S(A)$. Then
the following conditions are mutually equivalent{\rm:}
\begin{enumerate}
\item $\rho\in\Ext(S(A))$.
\item There exists $x\in X$ such that $\rho(a)=(\Psi(a))(x)$ for all $a\in A$.
\item $\rho$ is multiplicative, i.e., $\rho(ab)=\rho(a)\rho(b)$ for all $a,b\in A$.
\item $\rho(p)\in\{0,1\}$ for all $p\in P$.
\end{enumerate}
\end{theorem}

\begin{proof}
By Theorem \ref{th:FnRepGH}, $C(X,\reals)$ is a GH-algebra and the restriction
of $\Psi$ to $P$ is the Boolean isomorphism $\psi$ of $P$ onto $P(X,\reals)$
provided by Stone's representation theorem. Evidently, $\Psi$ induces an
affine bijection $\rho\leftrightarrow\gamma$ between $S(A)$ and $S(C(X,
\reals))$ according to $\gamma=\rho\circ\Psi\sp{-1}$ and $\rho=\gamma
\circ\Psi$. (The symbol $\circ$ denotes function composition.) Thus if
$\rho\leftrightarrow\gamma$, then $\rho$ is extremal iff $\gamma$ is extremal.
Thus, let $\rho\in S(A)$ and let $\gamma:=\rho\circ\Psi\sp{-1}$.

(i) $\Leftrightarrow$ (ii). By Lemma \ref{States&ProbMeas} (iv), part (i)
holds iff there exists $x\in X$ such that $\gamma=\gamma\sb{x}$, i.e., iff
$\rho=\gamma\sb{x}\circ\Psi$, and the latter condition is clearly equivalent
to (ii).

(ii) $\Rightarrow$ (iii). Assume (ii) and let $a,b\in A$. Then $\rho(ab)=
(\Psi(ab))(x)=(\Psi(a)\Psi(b))(x)=(\Psi(a))(x)(\Psi(b))(x)=\rho(a)\rho(b)$.

(iii) $\Rightarrow$ (iv). Assume (iii) and let $p\in P$. Then as $p=p\sp{2}$,
we have $\rho(p)=\rho(p\sp{2})=(\rho(p))\sp{2}$, whence $\rho(p)\in\{0,1\}$.

(iv) $\Rightarrow$ (i). Assume (iv). By Lemma \ref{States&ProbMeas} (iv)
again, it will be sufficient to prove that there exists $x\in X$ with
$\gamma=\gamma\sb{x}$. For $p\in P$, $\rho(p)=\gamma(\Psi(p))$, and $\{\Psi(p):
p\in P\}=P(X,\reals)=\{\chi\sb{K}:K\in{\mathcal F}(X)\}$, whence (iv) is
equivalent to the condition that $\gamma(\chi\sb{K})\in\{0,1\}$ for all
$K\in{\mathcal F}(X)$. By Lemma \ref{States&ProbMeas} (ii), there exists a
probability measure $\mu$ on $X$ such that $\gamma(f)=\int\sb{X}f\,d\mu$ for
all $f\in C(X,\reals)$, and therefore $\mu(K)=\int\sb{X}\chi\sb{K}\,d\mu=
\gamma(\chi\sb{K})\in\{0,1\}$ for all $K\in{\mathcal F}(X)$.

For each $x\in X$, there are two possible cases. \emph{Case I}: If $x\in
K\in{\mathcal F}(X)$, then $\mu(K)\not=0$, whence $\mu(K)=1$. \emph{Case
II}: There exists $K\sb{x}\in{\mathcal F}(X)$ such that $x\in K\sb{x}$ and
$\mu(K\sb{x})=0$, whence $\mu(\{x\})=0$.

Suppose that $x\in X$, that \emph{Case I} holds for $x$, and that $U
\subset X$ is an open set with $x\in U$. Then there exists $K\in
{\mathcal F}(X)$ with $x\in K\subseteq U$ and $\mu(K)=1$, whence $\mu(U)
=1$. Therefore, since the measure $\mu$ is regular, $\mu(\{x\})=\inf
\{\mu(U):x\in U\subseteq X\text{\ and\ }U\text{\ is open}\}=1$. Such an
$x$ is unique, for if $y\in X$, $y\not=x$, and $\mu(\{y\})=1$, then
$\mu\{x,y\}=2$, contradicting $\mu(X)=1$.

We claim that \emph{Case I} holds for a (necessarily unique) $x\in X$.
Suppose not. Then \emph{Case II} holds for all $x\in X$, whence $X$ is
covered by the sets $K\sb{x}$, $x\in X$, with $\mu(K\sb{x})=0$. Since
each $K\sb{x}$ is open and $X$ is compact, it follows that there exist
$x\sb{1}, x\sb{2}, ..., x\sb{n}\in X$ such that $X=\bigcup\sb{i=1}\sp{n}K
\sb{x\sb{i}}$. But a finite union of sets of $\mu$-measure $0$ has
$\mu$-measure $0$, so $\mu(X)=0$, again contradicting $\mu(X)=1$.  Thus,
$\mu$ concentrates measure on a unique point $x\in X$, and it follows
that $\gamma(f)=\int\sb{X}f\,d\mu=f(x)$, i.e., $\gamma=\gamma\sb{x}$.
\end{proof}

\end{document}